\documentclass[12pt]{article}

\usepackage{amsfonts}
\usepackage{amssymb}
\usepackage{amsmath}
\usepackage{amsthm}
\usepackage{mathtools}
\usepackage{caption}
\usepackage{amsbsy}
\usepackage{cuted}
\usepackage{citesort}
\usepackage{booktabs}

\newtheorem{definition}{Definition}
\newtheorem{theorem}{Theorem}
\newtheorem{proposition}{Proposition}

\newtheorem{corollary}{Corollary}
\newtheorem{remark}{Remark}
\newtheorem{example}{Example}

\providecommand{\keywords}[1]{\textbf{\textit{Index terms---}} #1}

\begin{document}

\title{New Convolutional Codes Derived from Algebraic Geometry Codes}

\author{Francisco Revson F. Pereira\thanks{Francisco Revson F. Pereira (corresponding author) is with the Department
of Electrical Engineering, Federal University of Campina Grande (UFCG), 58429-900, Campina Grande, PB,
Brazil, e-mail: (francisco.pereira@ee.ufcg.edu.br).},
        Giuliano G. La Guardia\thanks{Giuliano G. La Guardia is with Department of Mathematics and Statistics,
State University of Ponta Grossa (UEPG), 84030-900, Ponta Grossa,
PR, Brazil, e-mail: (gguardia@uepg.br).},\\
        Francisco M. de Assis\thanks{Francisco M. de Assis is with the Department
of Electrical Engineering, Federal University of Campina Grande (UFCG), 58429-900, Campina Grande, PB,
Brazil, e-mail: (fmarcos@dee.ufcg.edu.br).}}


\maketitle

\begin{abstract}
In this paper, we construct new families of convolutional codes.
Such codes are obtained by means of algebraic geometry codes.
Additionally, more families of convolutional codes are constructed
by means of puncturing, extending, expanding and by the direct
product code construction applied to algebraic geometry codes. The
parameters of the new convolutional codes are better than or
comparable to the ones available in literature. In particular, a
family of almost near MDS codes is presented.
\end{abstract}


%\IEEEpeerreviewmaketitle

\section{Introduction}
\label{sec:Introduction}

The class of convolutional codes is a class of codes much
investigated in the literature
\cite{Forney:1970,Lee:1976,Piret:1988,Rosenthal:1999,Rosenthal:2001,Luerssen:2008}.
Constructions of convolutional codes with good parameters or even
maximum distance separable (MDS), i.e. optimal, convolutional codes
(in the sense that they attain the generalized Singleton bound
\cite{Rosenthal:1999}) have also been presented in the literature
\cite{Forney:1970,Lee:1976,Piret:1988,Rosenthal:1999,Rosenthal:2001,Luerssen:2008,LaGuardia:2014,LaGuardia:2014A,LaGuardia:2016}.
Rosenthal~\emph{et al.} introduced the generalized Singleton bound
\cite{Rosenthal:1999} (see also \cite{Rosenthal:2001}) in 1999.

In this paper, we construct several new families of unit-memory
convolutional codes derived from classical algebraic geometry (AG)
codes. To do this, we apply the method introduced by Piret
\cite{Piret:1988} which was generalized by Aly~\emph{et al.}
\cite{Aly:2007}. Additionally, we utilize the techniques of code
expansion, puncturing, extension and the product code construction
in order to obtain more families of convolutional codes. An
advantage of our constructions lies in the fact that the new
convolutional codes are generated algebraically and not by
computational search. Moreover, since there exist classical AG codes
with good parameters, our new convolutional codes also have good
parameters. The class of AG codes was introduced by Goppa
\cite{Goppa:1981} in 1981. These codes have nice properties and are
asymptotically good. There exist several works dealing with
investigations concerning algebraic geometry (AG) codes
\cite{Stichtenoth:2002,Jin:2010,Jin:2012,Jin:2014,Munuera:2016}.
However, only few papers \cite{Perez:2004,Martin:2013,Curto:2012}
address the construction of convolutional codes by applying AG codes
as their classical counterpart.

A natural question that can arise is as follows: why it is important
to obtain convolutional codes which are not MDS, since there exist
MDS codes? The answer is simple: MDS codes is known to exist for
specific code lengths constructed over specific alphabets. For
example, in
Refs.~\cite{LaGuardia:2014,LaGuardia:2014A,LaGuardia:2016}, one has
convolutional MDS codes of length $n=q+1$ or $n=\frac{(q+1)}{2}$ (in
the last case, $q\equiv 3 \operatorname{mod} 4$) over
$\mathbb{F}_q$. In Ref.~\cite{Gluesing:2006}, most of the codes are
constructed over large alphabets when compared to its code length.
Other example is Ref.~\cite{Klapp:2007}, where convolutional MDS
codes over $\mathbb{F}_q$ with code length $n| (q^2 - 1)$ and $q+1 <
n \leq q^2 - 1$ were constructed.

The paper is organized as follows. In
Section~\ref{sec:convolutionalCodes}, we review basic concepts on
convolutional codes. In Section~\ref{sec:AGCodes}, a review of
concepts concerning algebraic geometry codes is given. In
Section~\ref{sec:convAGCodes}, we propose constructions of new
families of convolutional codes. In particular, a family of almost
near MDS (or near MDS or MDS) convolutional codes is shown. In
Section~\ref{sec:codesComp}, we compare the new code parameters with
the ones shown in the literature. Finally, in
Section~\ref{finalrem}, the final remarks are drawn.

\section{Review of Convolutional Codes}
\label{sec:convolutionalCodes}
In this section we present a brief review of classical convolutional
codes. For more details we refer the reader to
\cite{Forney:1970,Piret:1988,Johannesson:1999,Rosenthal:2001,Huffman:2003,Aly:2007,Klapp:2007,Luerssen:2008,LaGuardia:2014}.

\emph{Notation.} Throughout this paper, $p$ denotes a prime number,
$q$ is a prime power, $\mathbb{F}_{q}$ is the finite field with $q$
elements and $F/\mathbb{F}_q$ is an algebraic functions field over
$\mathbb{F}_q$ of genus $g$.

We begin with a few usual definitions used in the theory of
convolutional codes. A polynomial encoder matrix $G(D) \in
\mathbb{F}_{q}{[D]}^{k \times n}$ is called \emph{basic} if exists a
polynomial right inverse for $G(D)$. A minimal-basic generator
matrix is a encoder matrix which the overall constraint length
$\gamma =\displaystyle\sum_{i=1}^{k} {\gamma}_i$ has the smallest
value among all basic generator matrices (in this case, the overall
constraint length $\gamma$ is called the \emph{degree} of the
corresponding code).

\begin{definition}\cite{Klapp:2007}
A rate $k/n$ \emph{convolutional code} $C$ with parameters $(n, k,
\gamma ;$ $m, d_{f} {)}_{q}$ is a submodule of $\mathbb{F}_q
{[D]}^{n}$ generated by a reduced basic matrix $G(D)=(g_{ij}) \in
\mathbb{F}_q {[D]}^{k \times n}$, that is, $C = \{ {\bf u}(D)G(D) |
{\bf u}(D)\in \mathbb{F}_{q} {[D]}^{k} \}$, where $n$ is the length,
$k$ is the dimension, $\gamma =\displaystyle\sum_{i=1}^{k}
{\gamma}_i$ is the \emph{degree}, where ${\gamma}_i = {\max}_{1\leq
j \leq n}$ $\{ \deg g_{ij} \}$, $m = {\max}_{1\leq i\leq
k}\{{\gamma}_i\}$ is the \emph{memory} and $d_{f}=$wt$(C)=\min
\{wt({\bf v}(D)) \mid {\bf v}(D) \in C, {\bf v}(D)\neq 0 \}$ is the
\emph{free distance} of the code.
\end{definition}

A generator matrix $G(D)$ is called \emph{catastrophic} if there
exists a information sequence ${\bf u}{(D)}^{k}\in
\mathbb{F}_{q}{((D))}^{k}$ of infinite Hamming weight such that
results in a codeword ${\bf v}{(D)}^{k} = {\bf u}{(D)}^{k}G(D)$ with
finite Hamming weight. Since a basic generator matrix is
non-catastrophic, the convolutional codes constructed in this paper
have non catastrophic generator matrices.

The Euclidean inner product of two vectors ${\bf u}(D) = {\sum}_i
{\bf u}_i D^i$ and ${\bf v}(D) = {\sum}_j {\bf u}_j D^j$ in
$\mathbb{F}_q {[D]}^{n}$ is defined as $\langle {\bf u}(D)\mid {\bf
v}(D)\rangle = {\sum}_i {\bf u}_i \cdot {\bf v}_i$. For a
convolutional code $C$, the Euclidean dual of $C$ is defined by
$C^{\perp }=\{ {\bf u}(D) \in \mathbb{F}_q {[D]}^{n}\mid \langle
{\bf u}(D)\mid {\bf v}(D)\rangle = 0$ for all ${\bf v}(D)\in C\}$.

Let ${[n, k, d]}_{q}$ be a linear code with parity check matrix $H$.
One first splits $H$ into $m+1$ disjoint submatrices $H_i$ such that
\begin{eqnarray}
H = \left[
\begin{array}{c}
H_0\\
H_1\\
\vdots\\
H_{m}\\
\end{array}
\right].
\end{eqnarray}
After this, we consider the polynomial generator matrix given by
\begin{eqnarray}
G(D) =  {\tilde H}_0 + {\tilde H}_1 D + {\tilde H}_2 D^2 + \ldots +
{\tilde H}_m D^m,
\end{eqnarray}
where the matrices ${\tilde H}_i$, for all $1\leq i\leq m$, are
derived from the respective matrices $H_i$ by adding zero-rows at
the bottom in such a way that the matrix ${\tilde H}_i$ has $\kappa$
rows in total, where $\kappa$ is the maximal number of rows among
all the matrices $H_i$. The matrix $G(D)$ generates a convolutional
code with memory $m$.

\begin{theorem}\cite[Theorem 3]{Aly:2007}\label{A}
Let $C \subseteq \mathbb{F}_q^n$ be a linear code with parameters
${[n, k, d]}_{q}$. Assume also that $H \in \mathbb{F}_q^{(n-k)\times
n}$ is a parity check matrix for $C$ partitioned into submatrices
$H_0, H_1, \ldots, H_m$ as in Eq.~(1) such that $\kappa =$ rk$H_0$
and rk$H_i \leq \kappa$ for $1 \leq i\leq m$ and consider the
polynomial matrix $G(D)$ as in Eq.~(2). Then the matrix $G(D)$ is a
reduced basic generator matrix. Additionally, if $d_f$ denotes the
free distances of the convolutional code $V$ generated by $G(D)$,
and $d^{\perp}$ is the minimum distance of $C^{\perp}$, then one has
$d_f \geq d^{\perp}$.
\end{theorem}

To finish this section, we recall the generalized Singleton bound
\cite{Rosenthal:2001} of an $(n, k, \gamma ; m, d_{f} {)}_{q}$
convolutional code, which says that the free distance is upper
bounded by $d_{f}\leq (n-k)[ \lfloor \gamma/k \rfloor +1 ] + \gamma
+1$.

\section{Review of Algebraic Geometry Codes}
\label{sec:AGCodes}
In this section, we introduce some basic notation and results of
algebraic geometry codes. For more details, the reader can see
\cite{Stichtenoth:2009,Tsfasman:2007}.

Let $F/\mathbb{F}_q$ be a algebraic functions field of genus $g$. A place $P$ of $F/\mathbb{F}_q$
is the maximal ideal of some valuation ring $\mathcal{O}$ of $F/\mathbb{F}_q$. We also define
$\mathbb{P}_F := \{P | P\textit{ is a place of }F/\mathbb{F}_q\}$. A divisor of $F/\mathbb{F}_q$ is a
formal sum of places given by
$D := \sum_{P\in \mathbb{P}_F} n_P P, \textit{ with }n_P\in\mathbb{Z}, \textit{ almost all } n_P=0$.
The support of $D$ is defined as $suppD:=\{P\in \mathbb{P}_F|n_p\neq 0\}$. The discrete valuation
corresponding to a place $P$ is written as $\nu_P$. For every element $x$ of $F/\mathbb{F}_q$, we can
define a principal divisor of $x$ by $(x) := \sum_P\nu_P (x) P$. For a given divisor $G$, we denote the
Riemann-Roch space associated to $G$ by $\mathcal{L}(G) = \{x \in F/K \setminus \{0\}| (x) \geq -G\}$.

Let $\Omega_F := \{\omega|\omega \textit{ is a Weil differential of
}F/K\}$ be the differential space of $F/\mathbb{F}_q$. Given a
nonzero differential $w$, we denote by $(\omega):=\sum_P \nu_P(w) P$
the canonical divisor. All canonical divisor are equivalent and have
degree equal to $2g-2$. Furthermore, for a divisor $A$ we define
$\Omega_F(G) := \{\omega \in \Omega_F| \omega = 0 \textit{ or }
(\omega)\geq G\}$, and denote its dimension by $i(G)$.

\begin{theorem}(Riemann-Roch Theorem)\cite[Theorem 1.5.15, pg 30]{Stichtenoth:2009}
Let $W$ be a canonical divisor of $F/K$. Then for each divisor $G$,
the dimension of $\mathcal{L}(G)$ is given by $\ell(G) =
\operatorname{deg} (G) + 1 - g + \ell(W - G),$ where $W$ is a
canonical divisor.
\end{theorem}

Let $P_1, \ldots, P_n$ be pairwise distinct places of
$F/\mathbb{F}_q$ of degree $1$ and $D = P_1 + \ldots + P_n$. Choose
a divisor $G$ of $F/\mathbb{F}_q$ such that $supp G\cap supp D =
\varnothing$. Then one has:

\begin{definition}\cite[Definition 2.2.1, pg 48]{Stichtenoth:2009}
The algebraic geometry (AG) code $C_{\mathcal{L}}(D, G)$ associated
with the divisors $D$ and $G$ is defined as $C_{\mathcal{L}}(D, G)
:= \{(x(P_1), \ldots, x(P_n))| x\in \mathcal{L}(G)\}.$
\end{definition}

\begin{proposition}\cite[Corollary 2.2.3, pg
49]{Stichtenoth:2009}\label{AG1} Let $F/\mathbb{F}_q$ be a function
field of genus $g$. Then the AG code $C_{\mathcal{L}}(D, G)$ is an
$[n, k, d]$-linear code over $\mathbb{F}_q$ with parameters $k =
\ell(G) - \ell(G-D)\textit{ and } d\geq n- \operatorname{deg} (G)$.
If $2g - 2 < \operatorname{deg}(G) < n$, then $k =
\operatorname{deg}(G) - g + 1$. If $\{x_1, \ldots, x_k\}$ is a basis
of $\mathcal{L}(G)$, then a generator matrix of $C_{\mathcal{L}}(D,
G)$ is given by
\begin{equation}
G_{\mathcal{L}} =
\begin{bmatrix}\
    x_1(P_1) & x_1(P_2) & \cdots & x_1(P_n) \\
    x_2(P_1) & x_2(P_2) & \cdots & x_2(P_n) \\
    \vdots  & \vdots  & \ddots & \vdots  \\
    x_k(P_1) & x_k(P_2) & \cdots & x_k(P_n)
\end{bmatrix}.
\end{equation}
\end{proposition}

\begin{definition}\cite[Definition 2.2.6, pg 51]{Stichtenoth:2009}
Let $G$ and $D=P_1 + \ldots + P_n$ be divisors as before. Then we
define the code by $C_\Omega (D, G) :=
\{(\textit{resp}_{P_1}(\omega), \ldots,$
$\textit{resp}_{P_n}(\omega)| \omega\in \Omega_F(G - D)\}$, where
$\textit{resp}_{P_i}(\omega)$ denotes the residue of $\omega$ at
$P_i$.
\end{definition}

\begin{proposition}\cite[Theorem 2.2.7, pg 51]{Stichtenoth:2009}\label{AG2}
Let $F/\mathbb{F}_q$ be a function field of genus $g$. Let $G$ and
$D=P_1+\ldots+P_n$ be divisors as before. If $2g - 2 <
\operatorname{deg}(G) < n$, then $C_\Omega (D, G)$ is an $[n, k',
d']$-linear code over $\mathbb{F}_q$, where $k' =
n+g-1-\operatorname{deg}(G)$ and $d' \geq \operatorname{deg}(G) -
(2g -2)$.
\end{proposition}

The relationship between the codes $C_{\mathcal{L}}(D,G)$ and $C_\Omega (D,G)$ is given in the next proposition.

\begin{proposition}\cite[Propositions 2.2.10 and 2.2.11, pg 54]{Stichtenoth:2009}
Let $\eta$ be a Weil differential such that $\nu_{P_i}(\eta) = -1$ and $\eta_{P_i} = 1$
for all $i=1, \ldots, n$. Then
$C_{\mathcal{L}} (D,G)^\bot = C_{\Omega} (D, G) = C_{\mathcal{L}} (D, D - G + (\eta))$,
where $C_{\mathcal{L}} (D, G)^\bot$ is the Euclidean dual of $C_{\mathcal{L}} (D, G)$.
\end{proposition}

%The End

\section{New Convolutional AG Codes}
\label{sec:convAGCodes}
In this section we present a general method to construct
convolutional codes from AG codes. More precisely, we obtain
convolutional codes whose generator matrix is derived from the AG
code $C_{\Omega}(D, G)$. We adopt the notation given in the last
section.

Our first result is given in the following:

\begin{theorem}\label{Theo1}
Let $F/\mathbb {F}_q$ be a function field of genus $g$. Consider the
AG code $C_{\Omega}(D, G)$ with $2g-2 < \operatorname{deg}(G) < n$,
where $\operatorname{deg}(G)$ is the degree of the divisor $G$. Then
there exists a unit-memory convolutional code with parameters $(n,
k-l, l; 1, d_f \geq d)_q$, where $l \leq k/2$,
$k=\operatorname{deg}(G) +1 - g $  and $d\geq
n-\operatorname{deg}(G)$.
\end{theorem}

\begin{proof}
Let us consider the AG code $C_{\Omega}(D, G)$ defined over
$F/\mathbb {F}_q$ with parity check matrix
\begin{equation}
H_{\Omega} =
\begin{bmatrix}\
    x_1(P_1) & x_1(P_2) & \cdots & x_1(P_n) \\
    x_2(P_1) & x_2(P_2) & \cdots & x_2(P_n) \\
    \vdots  & \vdots  & \ddots & \vdots  \\
    x_k(P_1) & x_k(P_2) & \cdots & x_k(P_n)
\end{bmatrix},
\end{equation} where $\{x_1, \ldots, x_k\}$ is a basis of $\mathcal{L}(G)$.
Let $C_{\mathcal{L}} (D, G)$ be the (Euclidean) dual of the code
$C_{\Omega}(D, G)$. A generator matrix of $C_{\mathcal{L}} (D, G)$
is equal to $H_{\Omega}$. We know that $C_{\mathcal{L}} (D, G)$ is
an AG code with parameters $[n, k=\operatorname{deg}(G) +1 - g,$
$d\geq n-\operatorname{deg}(G)]_{q}$, where $n=
\operatorname{deg}(D)$. We will construct a convolutional code
derived from $C_{\Omega}(D, G)$ as follows.

Define a convolutional code with generator matrix $${\mathbb G}(D) =
H_0 + \widetilde{H}_1 D,$$ where $H_0$ is the submatrix of
$H_{\Omega}$ consisting of the $k-l$ first rows and
$\widetilde{H}_1$ is the matrix consisting of the last $l$ rows of
$H_{\Omega}$ by adding zero-rows at the bottom such that the matrix
$\widetilde{H}_1$ has $k-l$ rows in total. From hypothesis, it
follows that $\operatorname{rk}{H}_0 \geq \operatorname{rk}
\widetilde{H}_1$. From Theorem~\ref{A}, the matrix ${\mathbb G}(D)$
is a reduced basic matrix. The convolutional code generated by
${\mathbb G}(D)$ is a unit-memory code with dimension $k-l$, degree
$l$ and free distance $d_f$. From Theorem~\ref{A}, it follows that
$d_f \geq d$. Therefore, there exist convolutional codes with
parameters $(n, k-l, l; 1, d_f)$, with $d_f \geq d$.
\end{proof}

\begin{remark}
It is interesting to note that Theorem~\ref{Theo1} can be easily
generalized by considering multi-memory convolutional codes.
However, since unit-memory convolutional codes always achieve the
largest free distance among all codes of the same rate (see
\cite{Lee:1976}) we restrict ourselves to the construction of
unit-memory codes.
\end{remark}

\begin{corollary}\label{Cor1}
Assume that all the hypotheses of Theorem~\ref{Theo1} hold. Then
there exists a convolutional code with parameters $(n, k-1, 1; 1,
d_f \geq d)_q$, where $k=\operatorname{deg}(G) +1 - g $  and $d\geq
n-\operatorname{deg}(G)$.
\end{corollary}

\begin{proof}
It suffices to consider $l = 1$ in Theorem~\ref{Theo1}.
\end{proof}

\begin{remark}
Note that in Corollary~\ref{Cor1}, it follows from the generalized
Singleton bound, that the free distance of the convolutional codes
constructed here are bounded by $d_f \leq n - k + 3$ (where $n$ and
$k$ are the parameters of $C_{\mathcal{L}} (D, G)$). Furthermore,
$d_f \geq n - deg(G) = n - (k + g - 1) = n - k + 1 - g$; so the free
distance $d_f$ is bounded by $n-k + 1 -g \leq d_f \leq n-k + 3$. In
particular, for function fields $F/\mathbb{F}_q$ with $g = 0$ the
new convolutional codes have free distance bound by $n-k + 1 \leq
d_f \leq n-k + 3$. In this case, observe that these codes are almost
near MDS or near MDS or MDS. In other words, the Singleton defect is
at most two. Therefore, we have constructed good families of
convolutional codes.
\end{remark}

\begin{corollary}\label{Cor2}
Let $F = \mathbb{F}_q(z)$ be a rational function field. For $\beta
\in \mathbb{F}_q$, let $P_{\beta}$ be the zero of $z - \beta$ and
denote by $P_{\infty}$ the pole of $z$ in $\mathbb{F}_q (z)$. Then
there exists a convolutional code with parameters $(q, r, 1; 1, d_f
\geq q - r)_q$, where $1 < r \leq q-1$.
\end{corollary}

\begin{proof}
Consider the AG code $C_{\mathcal{L}}(D, G)$ with $D =
\sum_{\beta\in \mathbb{F}_q} P_\beta$ and $G = r P_{\infty}$, where
$1 < r \leq q-1$. We know that $C_\mathcal{L}(D, G)$ has parameters
$n = q$, $k = r+1$ and $d \geq n - r$. Applying Corollary~\ref{Cor1}
to the AG code ${C_\mathcal{L}(D, G)}^{\perp}$, one can get
convolutional codes with the desired parameters.
\end{proof}

\begin{theorem}\label{Theo3}
Let $q=2^t$, where $t\geq 1$ is an integer. Then there exists an
$(2q^2, m - q/2, 1; 1, d_f \geq 2q^2 - m)_q$ convolutional code,
where $q-2 < m < 2q^2$.
\end{theorem}
\begin{proof}
It follows from the fact that in the function field $F =
\mathbb{F}_q(x,y)$, defined by the equation $y^2 + y = x^{q+1}$, it
is possible to construct an AG code with parameters $[2q^2, m - q/2
+ 1, d \geq 2q^2 - m]_q$, with $q-2 < m < 2q^2$ (see
\cite{Stichtenoth:1988,Jin:2014}).
\end{proof}

\begin{example}
Applying Theorem~\ref{Theo3} we can construct an $(32, 15, 1; 1. d_f
\geq 15)_4$ new convolutional code whose parameters are better than
the $(32, 15, 10; \mu, d_f \geq 9)_3$ code, shown in
\cite{LaGuardia:2013}, and better than the $(32, 16, \gamma; 1, d_f
\geq 5)_3$ code, shown in \cite{Aly:2007}. Additionally, our new
$(128, 64, 1; 1, d_f \geq 60)_8$ code is better than the $(128, 64,
35; \mu, d_f \geq 17)_7$ code, shown in \cite{LaGuardia:2013}, and
better that the $(128, 64, \gamma; 1,$ $d_f \geq 8)_7$ code, shown
in \cite{Aly:2007}.
\end{example}

\begin{theorem}\label{Theo4}
Let $q=2^t$, where $t\geq 1$ is an odd integer. Then there exists an
$(3q^2 - 2q, m - q + 1, 1; 1, d_f \geq 3q^2 - 2q-m)_q$ convolutional
code, where $2q - 4 < m < 3q^2 - 2q$.
\end{theorem}

\begin{proof}
Let $F$ be the function field over $\mathbb{F}_{q^2}$ defined by the
equation
\begin{equation*}
y^q + y = x^3.
\end{equation*}
The genus of $F$ equals $g = q-1$ and the number of rational places
(place of degree one) is equal to $3q^2 - 2q + 1$ (see
\cite{Jin:2014}). Let $D = P_1 + \ldots + P_{n}$ be a divisor, where
$n = 3q^2 - 2q$, and $G = mP_{3q^2 - 2q + 1}$, with $2g-2 < m < n$,
where $\{P_1, \cdots, P_{3q^2 - 2q + 1}\}$ are all pairwise distinct
rational places. Consider the AG code $C_{\mathcal{L}}(D,G)$; the
parameters of $C_{\mathcal{L}}(D,G)$ are $[n = 3q^2 - 2q, k = m + 1
- g, d \geq n - m]_q$, where $2q - 4 < m < 3q^2 - 2q$.

Applying Corollary~\ref{Cor2}, we can get an $(3q^2 - 2q, m - q + 1,
1; 1, d_f \geq 3q^2 - 2q - m)_q$ convolutional code, where $2q - 4 <
m < 3q^2 - 2q$.
\end{proof}

The next results are obtained from Theorem~\ref{Theo1} when
considering puncturing, extending, expanding and the product code
construction to AG codes.

\begin{theorem}
Assume the same notation of Theorem~\ref{Theo1}, and suppose that
$C_{\mathcal{L}} (D, G)$ has no minimum weight codeword with a
nonzero $j$-th coordinate. Then there exists an $(n - 1, k - l, l;
1, d_f)_q$ convolutional code, where $d_f \geq d$,
$k=\operatorname{deg}(G) +1 - g $, $l \leq k/2$ and $d\geq
n-\operatorname{deg}(G)$.
\end{theorem}

\begin{proof}
Let $C_{\mathcal{L}} (D, G)$ be the ${[n, k, d]}_{q}$ AG code
considered in Theorem~\ref{Theo1}, where $D = P_1 + \ldots + P_n$.
Now, let $D' = D - P_j$, where $j \in \{1, 2, \ldots, n\}$. We
define the puncture code $C_{\mathcal{L}} (D^{'}, G)$ derived from
$C_{\mathcal{L}} (D, G)$, which is also an AG code (see
\cite{Pellikaan:1991}).  Note that the supports of $D'$ and $G$ are
disjoint, i.e. the definition of $C_{\mathcal{L}} (D^{'}, G)$ makes
sense. From hypothesis (see \cite[Theorem1.5.1, pg
13]{Huffman:2003}), $C_{\mathcal{L}} (D^{'}, G)$ has parameters
$[n-1, k, d]_{q}$. Applying the same construction shown in
Theorem~\ref{Theo1} we can construct a convolutional code $V$ with
parameters $(n - 1, k-l, l; 1, d_f)_q$, where $d_f \geq d$. A
generator matrix ${\mathbb G}^{*} (D)$ for $V$ is given by

%\begin{figure*}[h!]
%\begin{strip}

\tiny
\begin{equation*}
\hspace{-2.5cm}{\mathbb G}^{*} (D) =
\begin{bmatrix}
    x_1(P_1) + x_{k-l+1}(P_1)D & x_1(P_2) + x_{k-l+1}(P_2)D & \cdots & x_1(P_{j-1}) + x_{k-l+1}(P_{j-1})D & x_1(P_{j+1}) + x_{k-l+1}(P_{j+1})D & \cdots & x_1(P_n) + x_{k-l+1}(P_n)D\\
    x_2(P_1) + x_{k-l+2}(P_1)D & x_2(P_2) + x_{k-l+2}(P_2)D & \cdots & x_2(P_{j-1}) + x_{k-l+2}(P_{j-1})D & x_2(P_{j+1}) + x_{k-l+2}(P_{j+1})D & \cdots & x_2(P_n) + x_{k-l+2}(P_n)D\\
    \vdots       & \vdots       & \ddots & \vdots           & \vdots           & \ddots & \vdots\\
    x_l(P_1) + x_{k}(P_1)D & x_l(P_2) + x_{k}(P_2)D & \cdots & x_l(P_{j-1}) + x_{k}(P_{j-1})D & x_l(P_{j+1}) + x_{k}(P_{j+1})D & \cdots & x_l(P_n) + x_{k}(P_n)D\\
    x_{l+1}(P_1) & x_{l+1}(P_2) & \cdots & x_{l+1}(P_{j-1}) & x_{l+1}(P_{j+1}) & \cdots & x_{l+1}(P_n)\\
    \vdots       & \vdots       & \ddots & \vdots           & \vdots           & \ddots & \vdots\\
    x_{k-l}(P_1) & x_{k-l}(P_2) & \cdots & x_{k-l}(P_{j-1}) & x_{k-l}(P_{j+1}) & \cdots & x_{k-l}(P_n)
\end{bmatrix}.
\end{equation*}
%\end{strip}
\end{proof}

\begin{theorem}
Assume the same notation of Theorem~\ref{Theo1}. Then there exists
an $(n+1, k-l, l; 1, d_f\geq d^e)_q$ convolutional code, where $d^e
= d$ or $d^e = d+1$, where $k=\operatorname{deg}(G) +1 - g $, $l
\leq k/2$ and $d\geq n-\operatorname{deg}(G)$.
\end{theorem}

\begin{proof}
Let us consider $C_{\mathcal{L}} (D, G)$ be the ${[n, k, d]}_{q}$ AG
code considered in Theorem~\ref{Theo1}. We construct a new code
$C_{\mathcal{L}}^{e} (D, G)$ by extending the code
$C_{\mathcal{L}}(D, G)$. This new code have parameters $[n+1, k,
d^e]_{q}$, with $d^e = d$ or $d^e = d+1$. Applying the method
utilized in the proof of Theorem~\ref{Theo1}, one can get an $(n+1,
k-l, l; 1, d_f \geq d^e)_q$ convolutional code, and the result
follows.
\end{proof}

\begin{theorem}
Assume the same notation of Theorem~\ref{Theo1}. Then there exists
an $(mn, mk-l, l; 1, d_f\geq d)_q$ convolutional code, where
$k=\operatorname{deg}(G) +1 - g $, $l \leq k/2$ and $d\geq
n-\operatorname{deg}(G)$.
\end{theorem}

\begin{proof}
Consider that $C_{\mathcal{L}} (D, G)$ is the AG code, over
$\mathbb{F}_{q^m}$, with parameters ${[n, k, d]}_{q^{m}}$. Let
$\beta = \{b_1, \ldots, b_m\}$ be a basis of ${\mathbb F}_{q^m}$
over ${\mathbb F}_{q}$. We expand the code $C_{\mathcal{L}} (D, G)$
with respect of basis $\beta$ generating the code
$\beta(C_{\mathcal{L}} (D, G))$, over $\mathbb{F}_{q}$, with
parameters ${[mn, mk, d^{*}\geq d]}_{q}$. A parity check matrix $H$
of $[\beta(C_{\mathcal{L}} (D, G))]^{\perp}$ is a generator matrix
of $\beta(C_{\mathcal{L}} (D, G))$.

Let $V$ be the convolutional code generated by the minimal-basic
matrix
\begin{equation}
\mathbb{G}(D) = H_0 + \tilde{H}_1 D,
\end{equation}
where $H_0$ is a submatrix of $H$ consisting of the $mk-l$ first
rows of $H$ and $\tilde{H}_1$ is the matrix consisting of the last
row of $H$ and more $mk-2l$ zero rows. Then, we construct a
convolutional code $V$ that has parameters ${(mn, mk - l, l; 1, d_f\geq
d)}_{q}$, as desired.
\end{proof}

\begin{theorem}
Assume the same notation of Theorem~\ref{Theo1}. Then there exists
an $(n^2, k^2-l, l; 1, d_f\geq d^2)_q$ convolutional code, where
$k=\operatorname{deg}(G) +1 - g $, $l \leq k/2$ and $d\geq
n-\operatorname{deg}(G)$.
\end{theorem}

\begin{proof}
Let $C_{\mathcal{L}} (D, G)$ be the AG code of Theorem~\ref{Theo1}
over $\mathbb{F}_{q}$, with parameters ${[n, k, d]}_{q}$. We
construct a product code $(C_{\mathcal{L}} (D, G)\otimes
C_{\mathcal{L}} (D, G))$. This is an $[n^2, k^2, d^2]_q$ code.
Similarly to the proof of Theorem~\ref{Theo1}, one has an $(n^2, k^2
- l, l; 1, d_f\geq d^2)$ convolutional code, as required.
\end{proof}

%The End

\section{Code Comparisons}
\label{sec:codesComp}

In this section, we compare the parameters of the new convolutional
codes with the ones available in the literature. Table~\ref{Tab1},
shows a family of almost near MDS (or near MDS or MDS) codes
constructed from Corollary~\ref{Cor2}.

\begin{table}[h!]
    \centering
    \caption{New almost near MDS or near MDS or MDS codes}
    \begin{tabular}{|c|}
    \hline
    The new codes from Corollary~\ref{Cor2} \\\hline
    $(n, k, \gamma; m, d_f)$ \\ [0.5ex]\hline\hline
    $(8, 2, 1; 1, d_f\geq 6)_8$    \\\hline
    $(8, 5, 1; 1, d_f\geq 3)_8$    \\ [0.5ex]\hline\hline
    $(37, 17, 1; 1, d_f\geq 20)_{37}$    \\\hline
    $(37, 33, 1; 1, d_f\geq 4)_{37}$    \\ [0.5ex]\hline\hline
    $(71, 35, 1; 1, d_f\geq 36)_{71}$    \\\hline
    $(71, 68, 1; 1, d_f\geq 3)_{71}$    \\ [0.5ex]\hline\hline
    $(128, 64, 1; 1, d_f\geq 64)_{128}$    \\\hline
    $(128, 125, 1; 1, d_f\geq 3)_{128}$    \\ [0.5ex]\hline\hline
    $(256, 128, 1; 1, d_f\geq 128)_{256}$    \\\hline
    $(256, 253, 1; 1, d_f\geq 3)_{256}$    \\ [1ex] \hline
    \end{tabular}
    \label{Tab1}
\end{table}

The codes displayed in Table~\ref{Tab2} are obtained from
Theorems~\ref{Theo3}~and~\ref{Theo4}. Note that these new $(32, 15,
1; 1, d_f \geq 15)_4$ convolutional code is better than the $(32,
15, 10; \mu, d_f \geq 9)_3$ and $(32, 16, \gamma; 1, d_f \geq 5)_3$
shown in Refs. \cite{LaGuardia:2013} and \cite{Aly:2007},
respectively. The new $(128, 64, 1; 1, d_f \geq 60)_8$ code is
better than the $(128,64,35;$ $\mu, d_f \geq 17)_7$ and the $(128,
64, \gamma; 1, d_f \geq 8)_7$ from
Refs.~\cite{LaGuardia:2013}~and~\cite{Aly:2007}, respectively.

The other new codes shown in Table~\ref{Tab2} have different
parameters when compared to the ones available in literature.
Because of this fact, it is not possible to compare such codes with
the ones available in the literature.

\begin{table*}[h!]
    \centering
    \small
    \caption{Code Comparison}
    \begin{tabular}{|c |c |c|}
    \hline
    New codes & Codes in \cite{LaGuardia:2013} & Codes in \cite{Aly:2007} \\ [0.5ex]\hline\hline
    $(32, 15, 1; 1, d_f \geq 15)_4$  & $(32, 15, 10; \mu, d_f \geq 9)_3$   & $(32, 16, \gamma; 1, d_f \geq 5)_3$ \\\hline
    $(32, 1, 1; 1, d_f \geq 30)_4$  & --   & -- \\[0.5ex]\hline\hline
    $(128, 64, 1; 1, d_f \geq 60)_8$ & $(128, 64, 35; \mu, d_f \geq 17)_7$ & $(128, 64, \gamma; 1, d_f \geq 8)_7$ \\\hline
    $(176, 64, 1; 1, d_f \geq 105)_8$  & --   & -- \\\hline
    $(128, 3, 1; 1, d_f \geq 122)_8$  & --   & -- \\\hline
    $(176, 6, 1; 1, d_f \geq 163)_8$  & --   & -- \\[0.5ex]\hline\hline
    $(512, 128, 1; 1, d_f \geq 376)_{16}$ & -- & -- \\\hline
    $(512, 256, 1; 1, d_f \geq 248)_{16}$  & --   & -- \\[0.5ex]\hline\hline
    $(2048, 1024, 1; 1, d_f \geq 1008)_{32}$ & -- & -- \\\hline
    $(3008, 1024, 1; 1, d_f \geq 1953)_{32}$  & --   & -- \\\hline
    $(2048, 15, 1; 1, d_f \geq 2017)_{32}$  & --   & -- \\\hline
    $(3008, 30, 1; 1, d_f \geq 2947)_{32}$  & --   & -- \\ [1ex]
\hline
\end{tabular}
\label{Tab2}
\end{table*}

\newpage

\section{Final Remarks}\label{finalrem}
In this paper we have constructed new families of convolutional
codes derived from algebraic geometry codes. These new codes have
good parameters. More precisely, a family of almost near MDS codes
was presented. Additionally, our codes are better than or comparable
to the ones shown in \cite{Aly:2007,LaGuardia:2013}. Furthermore,
more families of convolutional codes were constructed by means of
puncturing, extending, expanding and by the direct product code
construction applied to algebraic geometry codes.

\section*{Acknowledgment}

This research has been partially supported by the Brazilian Agencies CAPES and
CNPq.

%\ifCLASSOPTIONcaptionsoff
%  \newpage
%\fi

\end{document}